\documentclass[sigplan, screen]{acmart}
\AtBeginDocument{%
  }

\copyrightyear{2026}
\acmYear{2026}
\setcopyright{cc}
\setcctype{by}
\acmConference[HSCC '26]{29th ACM International Conference on Hybrid Systems: Computation and Control}{May 11--14, 2026}{Saint Malo, France}
\acmBooktitle{29th ACM International Conference on Hybrid Systems: Computation and Control (HSCC '26), May 11--14, 2026, Saint Malo, France}
\acmDOI{10.1145/3801146.3805671}
\acmISBN{979-8-4007-2566-1/2026/05}

\usepackage{amsmath}
\usepackage{url}

\theoremstyle{plain}
\newtheorem{theorem}{Theorem}[section]

\newtheorem{lemma}[theorem]{Lemma}
\newtheorem{corollary}[theorem]{Corollary}
\theoremstyle{definition}
\newtheorem{definition}[theorem]{Definition}
\newtheorem{assumption}[theorem]{Assumption}
\theoremstyle{remark}
\newtheorem{remark}[theorem]{Remark}

\usepackage{tikz}
\usetikzlibrary{calc}
\usetikzlibrary{arrows.meta, positioning}

\usepackage{xcolor}
\usepackage{booktabs}
\usepackage{graphicx}
\usepackage{comment}
\usepackage{braket}
\usepackage{multirow}
\usepackage{caption}
\usepackage{subcaption}
\usepackage{stmaryrd}
\usepackage{dblfloatfix}
\usepackage{makecell}
\usepackage{mathtools}

\newcommand{\Sadegh}[1]{{\color{magenta} \textbf{SS:} #1}}

\def\reals{\mathbb{R}}

\def\nats{\mathbb{N}}

\def\M{\mathcal{M}}

\def\Sh{{\hat S}}
\def\xh{{\hat x}}

\def\sigmahat{{\hat \Sigma}}
\def\encoder{{\Lambda_e}}
\def\decoder{{\Lambda_d}}

\def\ball{{\mathbb B}}
\def\xinf{{x_\infty}}
\def\Dinf{{D_\infty}}
\def\piinf{{\pi_\infty}}
\def\tick{{\mathsf{T}}}

\begin{document}

\author{Rupak Majumdar}
\affiliation{%
  \institution{Max Planck Institute for Software Systems}
  \country{Germany}}
\email{rupak@mpi-sws.org}

\author{Mahmoud Salamati}
\affiliation{%
  \institution{Max Planck Institute for Software Systems}
  \country{Germany}}
\email{msalamati@mpi-sws.org}

\author{Nikhil Singh}
\affiliation{%
  \institution{Max Planck Institute for Software Systems}
  \country{Germany}}
\email{niksingh@mpi-sws.org}

\author{Sadegh Soudjani}
\affiliation{%
  \institution{Max Planck Institute for Software Systems, Germany}
  \country{University of Birmingham, United Kingdom}
  }
\email{sadegh@mpi-sws.org}


\title{Learning Metastable Dynamics} 



    \begin{abstract}

\emph{Metastability}---a phenomenon where systems remain trapped in quasi-stable states before abruptly transitioning under rare perturbations---is ubiquitous in physical systems. Although metastability is a widely observed phenomenon, its identification and analysis present significant challenges. 
To address these challenges, we propose a novel framework for analyzing metastability using {K}oopman theory. We use a finite set of system trajectories to learn a representation of the dynamics that defines a latent space in which the system evolves linearly, thereby enabling a systematic characterization of metastable behavior through the spectral properties of the linear mapping. Empirical evaluations demonstrate that our approach is capable of anticipating metastable behavior significantly earlier than its actual manifestation, even with $10\%$ of the simulation duration. Moreover, we establish that the dominant eigenvalue of the learned {K}oopman matrix in the latent space serves as a critical indicator for detecting metastability across both single-server and multi-server configurations.

	\end{abstract}

\ccsdesc[500]{Computing methodologies~Simulation tools}
\ccsdesc[500]{Computing methodologies~Modeling methodologies}

\keywords{Metastability, Koopman Theory, Queuing Systems, Performance Evaluation}


\maketitle


	
\section{Introduction}\label{sec:intro}

Metastability is a
widespread phenomenon in physical systems \cite{StochasticThermodynamics,MolecularDrivingForces,FW}, such as magnetic hysteresis~\cite{Baranov2019magnetic} or condensation of over-saturated water vapor~\cite{Cho2014vapourcondense}, where a system remains “persistently” in one state and then rapidly transitions into another in the presence of some rare event. Metastability plays a crucial role in the design and operation of distributed server systems—the backbone of today’s cloud computing infrastructure~\cite{Huang2022systemmetastability}. 

A representative instance of metastable behavior in distributed systems is a retry storm at a server. Retries are a mechanism to deal with failures: if a request is not responded to within a certain timeout,
something went wrong and the client is advised to retry the request. While retries are an excellent
mechanism to mitigate transient failures, in rare occasions, they may form a sustaining effect:
the additional workload from retries prevents the system to respond to requests on time, thereby
leading to further client-side retries that increases the workload. In the worst case, 
the retry storm
propagates to multiple services,
 culminating in large-scale service degradation or complete loss of availability.

In metastable systems, the interplay of time scales is of paramount importance. A system may be stable, yet exhibit extremely long transients, meaning that trajectories starting from certain regions of the state space take a very long expected time to \emph{settle down}.
Detecting such behavior is computationally intensive as it necessitates the simulation of numerous long trajectories. 
Moreover, even extensive simulation-based analysis often fails to reliably capture or characterize metastable phenomena due to the rarity and complexity of the underlying dynamics.

Existing studies on metastability have predominantly concentrated on linear systems or relatively simple models such as Markov chains. In case of Markov chains, metastability can be investigated through spectral analysis of the transition probability matrix~\cite{MetastabilityBook}.
However, for stochastic nonlinear dynamical systems, there remains no general framework for systematic metastability analysis.
A further challenge arises in many real-world engineering domains, where suitable analytical models are often unavailable for such analysis. This motivates the need to develop a general framework capable of learning arbitrarily complex dynamics from trajectories of the real system or its simulators.

In this work, we present the first practical methodology for analyzing metastability in stochastic nonlinear  systems with unknown dynamics. Our approach comprises three key components
\begin{itemize}
    \item \textit{Reduction to deterministic dynamics:}   Assuming ergodicity, we demonstrate that metastability in the underlying stochastic system can be studied through the time-scale separation in the corresponding \emph{deterministic} dynamics of statistical moments (e.g., expectations); \item \textit{Learning Koopman autoencoders:} By exploiting Koopman theory, we use a finite number of system trajectories to learn a representation of the deterministic dynamics that induces a latent space in which the system evolves linearly; and 
    \item \textit{Spectral characterization: } We characterize metastability by analyzing the spectral gap of the learned linear operator, establishing a rigorous connection between metastability properties in the latent and original state spaces. 
\end{itemize}

Our experimental study focuses on \textit{queuing} systems with retry mechanisms. The empirical results demonstrate that the proposed method can anticipate the onset of metastable behavior substantially earlier than its observable manifestation—requiring only approximately $10\%$ of the total simulation duration to do so. Furthermore, we show that the dominant eigenvalue of the learned Koopman matrix in the latent space functions as a robust and interpretable indicator of metastability, consistently identifying transitions across both single-server and multi-server configurations.
   In addition, we show that the theoretically derived upper bounds on the settling time align closely with the empirically observed settling times in both the original and latent state spaces, thereby validating the analytical consistency of our framework.
   Through our experimental analysis, we further investigate effective mitigation strategies for preventing metastable failures, including the implementation of controlled throttling mechanisms and the substitution of  queues with stacks in conventional server architectures.

    \section{Related Work}
\noindent\textbf{Metastability in the sciences.}
Metastability has been studied formally in the context of perturbed dynamical systems \cite{FW}, magnetic hysteresis~\cite{Baranov2019magnetic} or the condensation of supersaturated water vapor~\cite{Cho2014vapourcondense}.
Within the framework of Markov processes, Bovier et al.~\citep{BEGK1,BEGK2002,MetastabilityBook} defined metastability for discrete-time Markov chains and characterize metastability using potential theory and spectral methods. The application of spectral methods to metastability originates from the seminal works of Davies~\cite{Davies1982a,Davies1982b,Davies1983}, whose results primarily pertain to reversible Markov chains — i.e., chains satisfying the detailed balance condition — with subsequent technical generalizations extending to non-reversible settings.
Building on these foundations,
Betz and Le Roux~\cite{Betz2016} studied metastability for perturbed Markov chains, considering the asymptotics of metastability as the perturbation parameter goes to zero \cite{Betz2016}.
Despite these significant advances, existing analyses remain limited in generality, often relying on restrictive assumptions such as asymptotic regimes or reversibility. In contrast, the present work aims to provide a unified framework that captures metastable behavior beyond these assumptions.

\noindent\textbf{Metastable failures in systems.}
Metastable failures are introduced in the context of computer systems by Bronson et al.~\citep{bronson2021metastable}, providing illustrative examples and informal definitions of this phenomenon. Further, they emphasize that such failures have been observed across various domains in systems and networking, often with different names, such as persistent congestion \citep{AWS2021_persistent},
retry storms \citep{Azure2021_storm}, or cascading failures \citep{cascade2016}.
Building upon this foundation, Huang et al.~\cite{Huang2022systemmetastability} performed an extensive empirical study, demonstrating that metastable failures represent a prevalent cause of service outages across major software organizations.
Their work further refines the informal conceptualizations proposed by Bronson et al.~\cite{bronson2021metastable} and empirically reproduces these failures to better understand their dynamics.
Despite these advances, practical analysis and mitigation of metastable failures remain challenging due to the complex feedback loops that underlie their behavior and the difficulty of reproducing such failures in controlled environments. 
This paper addresses these challenges by providing a framework for detecting and characterizing metastable behavior. Habibi et al.~\cite{habibi2023msfmodel} modeled metastable server systems using continuous-time Markov chains (CTMCs) and employed Monte Carlo simulations to analyze their metastability.
Alvaro et al.~\cite{alvaro2025formalanalysismetastablefailures} observed that \emph{ab initio} CTMCs fail to capture the true system dynamics and proposed a calibration method that leverages data trajectories to adjust the CTMC transition rates. They then used the spectral properties of the calibrated CTMC to reason about metastability in simple FIFO systems. Isaacs et al.~\cite{analyzing-metastable-faiures-hotos-2025} propose using an integrated ensemble of CTMCs, discrete-event simulators, and service emulations to perform CTMC calibration more efficiently. However, their approach remains limited to systems that can be modeled as CTMCs.
In this paper, we propose a more general model and demonstrate experimentally that it can capture richer dynamics, such as those arising in multi-server systems and implementations with controlled throttling mechanisms. Our choice of model---Koopman autoencoders---enables formal metastability analysis of complex distributed server systems.

\noindent\textbf{Koopman theory.} The Koopman theory was introduced by Koopman in his seminal paper~\cite{Koopman1931PNAS}. In recent years, numerous studies~\cite{Brunton2021ModernKoopman} have explored the integration of Koopman theory with machine learning methodologies across a wide range of applications.One of the earliest contributions employing Koopman-based autoencoders to learn embeddings for linear dynamical systems is presented in~\cite{Lusch18}.
The work by Meanti et al.~\cite{Meanti2023koopmanestimation} extends this line of research by
estimating Koopman operators with sketching to provably learn large scale dynamical systems.
Furthermore, an adaptive approach for learning kernel functions to approximate the spectral properties of the Koopman operator is proposed in~\cite{Takeishi2019spectralKoopman}.
While these studies offer valuable insights into different strategies for learning within the Koopman framework, they have not been investigated in the context of analyzing metastability, which is part of this work.


\section{Preliminaries}\label{sec:prelims}

\textbf{Notation.} We denote the set of natural, real, and non-negative real numbers
by $\nats$, $\reals$, and $\reals_{\ge 0}$, respectively. Further, let $\nats_0=\nats \cup \{0\}$.
We
use superscript $n > 0$ with $\reals$ to denote the Cartesian
product of $n$ copies of $\reals$. For a vector $a\in \reals^n$, the notation $\|a\|_\infty$ is used for the infinity norm of $a$.
For $n\in \nats$, $\mathbf{1}_n$ denotes an $n$-dimensional vector whose entries are all one. Given $c\in \reals^n$ and $\delta \in \reals_{\ge 0}$, we define $\ball_\delta(c)\coloneqq \set{x\in \reals^n\mid \| x-c\|_\infty \leq \delta}$. We denote $i^{th}$ entry of $x\in\reals^n$ by $x[i]$. 
For a matrix $A\in \reals^{n\times n}$, we
denote the spectral radius of $A$ by $\rho(A)$, which is the largest modulus of the
eigenvalues of $A$. 
A matrix $A \in \mathbb{R}^{n \times n}$ is said to be \emph{diagonalizable} if there exists an invertible matrix $S$ and a diagonal matrix $\Gamma$ such that $A = S \Gamma S^{-1}$. The \emph{condition number} of $S$, defined (with respect to the $\ell_\infty$-norm) as
\[
\kappa_\infty \;=\; \|S\|_\infty \, \|S^{-1}\|_\infty.
\]
This quantity controls how the norm of $A^t$ can amplify relative to its spectral radius, since
\[
\|A^t\|_\infty \le \kappa_\infty \, \rho(A)^t.
\]

Let $\mu$ and $\nu$ be two probability measures on a measurable space 
$(S, \mathcal{B}(S))$. The \emph{total variation norm} (or \emph{total variation distance}) between 
$\mu$ and $\nu$ is defined as $\|\mu - \nu\|_{\mathrm{TV}} 
= \sup_{A \in \mathcal{B}(S)}|\mu(A) - \nu(A)|$. This represents the largest possible difference in probability that $\mu$ and $\nu$ assign 
to the same measurable event $A$. 


Let $A \subset \reals^n$ be a compact and connected set. $\ell_\infty$ radius of $A$ is an $n$-dimensional vector whose $i^{th}$ entry is defines as $\sup_{x,x'\in A}\|x[i]-x'[i]\|_\infty/2$. Let $A \subset \reals^n$ be a compact and connected set with $\ell_\infty$-radius greater than $\varepsilon > 0$. Then $A^{-\varepsilon}$ and $A^{+\varepsilon}$ denote the $\varepsilon$-shrunk and $\varepsilon$-expanded sets of $A$, respectively. Formally, $A^{-\varepsilon}=\set{x\mid \ball_\varepsilon(x)\subseteq A}$, and $A^{+\varepsilon}=\set{x\mid \ball_\varepsilon(x)\cap A \neq \emptyset}$. For two sets $A,B\subset \reals^n$, we define $\operatorname{dist}(A,B)=\inf_{(x,x')\in A\times B} \|x-x'\|_\infty$.

For a mapping $f\colon S \to \hat S$, with $S \subset \reals^n$ and $\hat S \subset \reals^m$, and a set $A \subseteq S$, we define $f(A)=\set{\xh\in \Sh\mid \exists x \in A,\; \xh=f(x)}$. 
Furthermore, we define
\begin{align}\label{eq:Lipschitz_const_def}
     U^f&\coloneqq \sup_{x,x'\in S} \frac{\|f(x')-f(x)\|_{\infty}}{\|x-x'\|_\infty}, \text{ and}\nonumber\\
     C^f&\coloneqq \inf_{x,x'\in S} \frac{\|f(x')-f(x)\|_{\infty}} {\|x-x'\|_\infty}.
     \end{align}

\subsection{Discrete-time Markov Processes}
\begin{definition}[Discrete-time Markov process]
Let $S \subseteq \mathbb{R}^n$ be a set, and let $\mathcal{B}(S)$ denote the Borel $\sigma$–algebra on $S$, i.e., the smallest $\sigma$–algebra containing all open subsets of $S$. 
A stochastic process $(X_t)_{t\in \nats_0}$ taking values in $S$ is called a 
\emph{discrete-time Markov process} 
if, for all $t \in \nats_0$ and all measurable sets $A\subseteq S$,
\[
\mathbb P\!\big(X_{t+1}\in A \,\big|\, X_0, X_1, \dots, X_t\big)
\;=\;
\mathbb P\!\big(X_{t+1}\in A \,\big|\, X_t\big)
\;=\;
P(X_t, A),
\]
where $P:S\times\mathcal B(S)\to[0,1]$ denotes the 
\emph{transition kernel} of the process.

\end{definition}


\begin{definition}[Stationary Distribution]
A probability measure $\piinf$ on $(S, \mathcal{B}(S))$ is called \emph{stationary} (or \emph{invariant}) if
\[
\piinf(A) = \int_{S} P(x, A)\, \piinf(dx), 
\quad \forall\, A \in \mathcal{B}(S).
\]
\end{definition}

\begin{definition}[Ergodicity]
The Markov process is said to be \emph{ergodic} if there exists a unique stationary probability measure $\piinf$, and for every initial state $x \in S$,
  \[
  \lim_{t\rightarrow \infty}\| P^t(x, \cdot) - \piinf(\cdot) \|_{\mathrm{TV}} = 0,
  \]
  where $P^t(x, \cdot)$ denotes the distribution of $X_t$ given $X_0=x$.
Ergodicity implies that $X_t$ converges to $\piinf$ in probability as $t \to \infty$, independently of the initial state $x$.
\end{definition}

\subsection{Dynamical Systems}

\begin{definition}[Deterministic dynamical systems]
\label{def:dynamical_sys}
A \emph{deterministic dynamical system} $\Sigma$ is defined as a tuple 
$\Sigma = (S, f)$, where $S \subseteq \mathbb{R}^n$ denotes the \emph{state space}, and
$f \colon S \to S$ specifies the \emph{dynamics}. 
The system evolves according to
\[
x_{t+1} = f(x_t), \quad x_0 \in S.
\]
If $f(x) = Ax$ with $A \in \reals^{n \times n}$, the system is said to be \emph{linear}. Such systems are denoted by $\Sigma = (S, A)$.
\end{definition}

For a linear dynamical system with state matrix $A$, we order the eigenvalues of $A$ in descending order of their modulus, such that $\lambda_1$ has the largest modulus and $\lambda_n$ the smallest. 

A discrete-time linear dynamical system $\Sigma=(S,A)$ is \emph{stable} if $\rho(A)<1$. For a bounded and connected state space $S\subset \reals^n$, we define the diameter of $S$ as follows:
\begin{equation}\label{eq:absolute_diameter}
    D(S) \coloneqq \sup_{x,x' \in S} \|x-x'\|_\infty.
\end{equation}

In the following, we introduce two notions for formalizing time scales in dynamical systems which will be used throughout of the paper.

\begin{definition}[$\varepsilon$-hitting time]\label{def:eps_hitting_time}
 For a given dynamical system $\Sigma=(S,f)$, constant $\varepsilon>0$ and sets $C,D\subset S$ with $\ell_\infty$ radius greater than $\varepsilon$ 
 and $\operatorname{dist}(C,D)>\varepsilon$, 
 $\tau_\varepsilon^\Sigma(C,D)$ denotes the infimum over the first time to reach $D^{+\varepsilon}$ when starting from $x_0\in C^{-\varepsilon}$.
\end{definition}

\begin{definition}[$\delta$-settling time]\label{def:delta_settling_time}
    For a dynamical system $\Sigma=(S,f)$ and a real valued constant $\delta>0$, let $x_0\in S$ and $\xinf\coloneqq \lim
    _{t\rightarrow \infty}x_t$ denote initial state and the corresponding steady state value. We define $\delta$-settling time as    \begin{equation}\label{eq:delta_settling_time_def}
    T_\delta^\Sigma \coloneqq \sup_{x_0\in S}\min \set{t\in \nats\mid \|x_t - \xinf \|_\infty<\delta}.
    \end{equation}

\end{definition}

\begin{remark}
For a stable linear system $\Sigma=(S,A)$ with $\rho(A)<1$, we have $x_\infty=0$ and
    \begin{equation}\label{eq:delta_settling_time_def_linear}
    T_\delta^\Sigma \coloneqq \sup_{x_0\in S}\min \set{t\in \nats\mid \|A^t x_0 \|_\infty<\delta}.
    \end{equation}
\end{remark}

\subsection{Koopman Autoencoders}

A \emph{Koopman autoencoder} model $\M =(\encoder, \decoder, A)$ consists of the following components:
\begin{enumerate}
    \item \textbf{Encoder} $\encoder$ is a mapping $\encoder: \reals^n \to \reals^m$ that maps the system's initial state 
    $x_0 \in \reals^n$ into a latent representation $\hat{x}_0 = \encoder(x_0)$.
    \item \textbf{Linear Koopman operator} $A \in \reals^{m \times m}$ models the linear evolution 
    of the latent dynamics via $\hat x_{t+1} = A \hat x_t$ for every $t\in \nats_0$.
    \item \textbf{Decoder} $\decoder$ is a mapping $\decoder: \reals^m \to \reals^n$ that reconstructs the original 
    state from its latent representation, i.e., $x_t' = \decoder(\hat{x_t})$ for every $t\in \nats_0$.
\end{enumerate}
\begin{assumption}
 Throughout the paper we assume that the Koopman operator $A$ is diagonalizable. 
\end{assumption}

\subsection{M/M/1 with Retries And Timeouts}

An M/M/1 queue models a simple client-server system with a single queue and a single server.
Clients send requests according to an exponential distribution with rate $\nu$.
Requests are enqueued at the tail and processed in a First In, First Out (FIFO) order.
Each request is served at an exponential rate $\mu$. The service times are independent from each other and from the arrival process. 


The M/M/1 model can be extended to incorporate timeouts and retry mechanisms.
\emph{Timeout} means that there is a constant $T_{timeout}$ such that, if a request has not been served within $T_{timeout}$, a client can take further action. 
This action can take one of two forms: a \emph{retry}, where a new instance of the request is added to the queue without removing the original request, or a \emph{drop}, where the client abandons the request entirely. Incorporating these behaviors allows the model to capture the dynamics underlying phenomena such as retry storms and metastable failures in queuing systems.

M/M/1 queues with retries and timeouts can be modeled as continuous-time Markov chains (CTMCs) defined over discrete state spaces~\cite{habibi2023msfmodel,alvaro2025formalanalysismetastablefailures}.
A fundamental limitation of CTMC-based approaches, however, is the requirement of a discrete state space whose size grows exponentially with the number of servers, resulting in very large CTMCs. This renders metastability analysis computationally intractable in practice, even when black-box linear algebra methods are utilized to exploit the existing sparsity. Moreover, extending CTMCs to systems with additional features---such as multiple servers, throttling mechanisms, or substituting FIFO queues with LIFO stacks---is nontrivial and, when possible at all, leads to further combinatorial explosion of the state space. 

Although we ultimately analyze discrete-time trajectories, these trajectories are obtained from a discrete-event simulator that operates in continuous time. Our goal is not to model the exact queue-length dynamics at the state-transition level, but rather to capture the aggregate or average behavior of the system.
Therefore, in this paper, we focus on discrete-time Markov processes with continuous state spaces, which can efficiently capture the dynamics of a broader class of real-world queuing systems.

\section{Metastability Characterization}\label{sec:metastability}
We consider stochastic systems exhibiting metastability due to rare events and show that, for ergodic systems, metastability admits a deterministic formulation in terms of the moment dynamics. We thus start with \emph{ergodic} stochastic systems. 

\begin{lemma}[Expectation convergence for ergodic Markov processes]\label{lem:expectation_convergence}
Let $(X_t)_{t\in \nats_0}$ be a discrete-time Markov process on a measurable state space
$(S,\mathcal B(S))$ with $S \subset \mathbb R^n$ \emph{bounded}.
Assume the Markov process is ergodic. Then, 
\[
\mathbb E_x\big[X_t\big]
\;=\;\int_{S}y\,P^t(x,dy)
\;\xrightarrow[t\to\infty]{}\;
\int_{S}y\,\piinf(dy),
\]
that is expectation of $X_t$ converges to a unique stationary (steady-state) value.
\end{lemma}
\begin{proof}
Since $S$ is bounded,
$g(y)=y$ is bounded on $S$:
there exists $M<\infty$ with $\sup_{y\in S}g(y) \le M$.
Hence $g$ is a bounded measurable function.

Total variation convergence implies convergence of expectations for
bounded measurable test functions: for any bounded measurable $g$,
\[
\left|\int g\,dP^t(x,\cdot)-\int g\,d\pi\right|
\;\le\; \|g\|_\infty\,\|P^t(x,\cdot)-\pi\|_{\mathrm{TV}}.
\]
Applying this we obtain
\[
\left|\mathbb E_x\big[X_t\big]-\int y\,\pi(dy)\right|
\;\le\; M\,\|P^t(x,\cdot)-\pi\|_{\mathrm{TV}}
\;\xrightarrow[t\to\infty]{}\;0.
\]
This proves the claim.
\end{proof}

\begin{remark}
The arguments used in the proof of the previous lemma extend directly to show that all higher-order moments of $X_t$ converge to unique steady-state values.
\end{remark}

Based on the lemma above, ergodicity guarantees that $\mathbb{E}(X_t)$ converges to a unique steady state value $\xinf\in S$. Hereafter, we consider $\Sigma=(S,f)$ as the dynamics that governs evolution of $\mathbb{E}(X_t)$ and use \begin{equation}\label{eq:average_state_notation}
x_t\coloneqq \mathbb{E}(X_t).
\end{equation}
We further denote the $\delta$-neighborhood of the attractor as 
\begin{equation}\label{eq:D_inf}
\Dinf\coloneqq \ball_\delta(\xinf).
\end{equation}



For stochastic systems, metastability is often characterized using either \emph{expected hitting time} that is required to take the system from one \emph{metastable set} to another, or using the \emph{mixing time}, that is the maximum time required for a process, starting from an arbitrary initial distribution, to reach the neighborhood of its stationary distribution. Systems with large mixing times are regarded as metastable \cite{alvaro2025formalanalysismetastablefailures}. In the following we provide an analogous characterization of metastability that is applicable to the average dynamics $\Sigma$. 
The characterization stated below uses the concept of $\varepsilon$-hitting time (Def.~\ref{def:eps_hitting_time}).


\begin{definition}[$(T, \varepsilon, \Dinf)$-metastability for ergodic Markov processes]\label{def:metastability_def_main}
 For an ergodic Markov process $(X_t)_{t\in \nats_0}$ defined over a measurable state space $
 (S, \mathcal B(S))$, let $\Sigma=(S,f)$ denote the dynamical system that describes the evolution of $x_t=\mathbb E(X_t)$.
Given 
\emph{candidate sets} $D_1, D_2, \dots, D_K\subset S$ ($K\geq 1$), attractor set $\Dinf\subset S$, time horizon $T \in \nats$, and constants $0<\varepsilon< \delta$, we say that $\Sigma$ is $(T, \varepsilon, D_\infty)$-metastable if the candidate sets are $\varepsilon$-separate from each other and
\begin{align}
    & \min_{1\leq i\leq K}\tau_\varepsilon^\Sigma(D_i, D_\infty)>T.
\end{align}
\end{definition}	
    \section{Metastability Analysis via Koopman Autoencoders}\label{sec:autoencoder_learning}

In Sec.~\ref{sec:metastability}, we provided formal characterizations for metastability in a deterministic dynamical system $\Sigma=(S,f)$ that corresponds to the \emph{average dynamics} of the underlying ergodic Markov process. The stated characterization relies on computation of $\tau_\varepsilon^\Sigma$ that in practice for general nonlinear $f$ may be computationally intractable. In this section, we (1) propose a modeling framework using Koopman autoencoders of the form $(\encoder,\decoder, A)$ for which we relate the metastability of $\Sigma$ and $\hat\Sigma=(\encoder(S), A)$, and (2) we provide quantitative relation between spectral radius of $\hat\Sigma$ and time scales of $\Sigma$.

\subsection{Learning from Data Trajectories}
In this section, we present a method for \emph{learning} a model that takes the form of Koopman autoencoders, using a finite set of trajectories generated either by the real system or a high-fidelity simulator. 
Our learning objective is to fit parameters of the Koopman autoencoder model with respect to a number of \emph{average trajectories} of the system.

To collect trajectories, we choose a set of initial states $\set{\bar x_0^{(1)}, \bar x_0^{(2)},\dots,\bar x_0^{(Z)}}$. 
For each $1\leq i \leq Z$, we run the simulator $M$ times, to produce $M$ simulated trajectories, each of length $L\in \nats$ and sampled regularly with respect to a chosen sampling time $\tick>0$. Note that $Z$, $M$, $L$, and $\tick$ are hyperparameters for the data collection.

For every $1\leq i \leq Z$ and $1\leq j \leq M$, we
write $\bar X^{i,j}_{t}$ for the simulator state at time steps $0\leq t \leq L-1$.
Note that $\bar X^{i,j}_{0}=\bar x_0^{(i)}$ for every $1\leq j \leq M$. 

Next, for every $1\leq i \leq Z$ and $0\leq t \leq L-1$, we  
compute the \emph{empirical} average 
\[
\bar x_t^i\coloneqq \frac{1}{M}\sum_{j=1}^M \bar X^{i,j}_t
\]
This gives the averaged dynamics of the simulator over $M$ runs.

We would like to ``match'' this average simulator dynamics to the output of the Koopman auto encoder model at corresponding times.
The corresponding auto encoder outputs are computed as 
\[
y_t^i(\theta)\coloneqq \decoder(\theta)(A^t(\theta)\encoder(\theta)(\bar x_0^i)),
\]
where $\theta$ denotes the set of parameters of the $\M(\theta)$.

Our aim is to find $\theta^\ast$ such that the average output trajectories of $\M(\theta^\ast)$ , i.e., $y^i_t({\theta^\ast})|_{t=0}^{L-1}$, match as closely as possible with the trajectories of the simulator, i.e., $\bar x_t^i|_{t=0}^{L-1}$, for every $1\leq i \leq Z$.

Formally, we solve the following optimization problem that minimizes the loss:
\begin{align}\label{eq:optimization_learning}
	&\min_{\theta}\sum_{i=1}^Z\sum_{t=0}^{L-1} \|y_t^i(\theta)- \bar x_t^i\|_2^2.
\end{align}

\subsection{Relating Time Scales}
We first consider the (unrealistic) case wherein the model introduces no error compared to the actual dynamics. 


\begin{definition}[Perfect representation] For a discrete-time deterministic dynamical system $\Sigma=(S,f)$ with $S\subset \reals^n$, let $\M=(\encoder,\decoder,A)$ denote a Koopman autoencoder model, 
such that for every $t\in \nats_0$ and $x_0 \in S$ we have $f^t(x_0)=\decoder(A^t \encoder(x_0))$. We then call $\M$ a perfect representation of $\Sigma$. 
\end{definition}


 The next lemma addresses the following question: 
under the perfect representation, can the metastability of the average dynamics 
$\Sigma$ be inferred from the metastability properties of the linear operator 
$A$ that constitutes part of the model $\M$?

\begin{lemma}[Time scale changes under perfect representation]\label{lem:perfect_representation}
Let $I,F\subset S$ be compact and connected sets, and let $\M=(\encoder,\decoder, A)$ be a perfect representation of $\Sigma=(S,f)$. 
Further, let $\hat I = \encoder(I)$ and $\hat F = \encoder(F)$. Then we have
\begin{equation}
\tau_0^\Sigma(I, F) = \tau_0^{\hat\Sigma}(\hat I, \hat F),
\end{equation}
where $\hat\Sigma=(A,\Sh)$ and $\Sh=\encoder(S)$.
\end{lemma}
\begin{proof}
Let $x_0, x_1,\ldots$ denote the sequence of states generated by applying $f$ consecutively. Let $N=\tau^\Sigma_0(x_0, F)$ for some $x_0 \in I$. For every $t<N$ we have $x_t \notin F$, and $x_N\in F$. This gives (1) $\encoder(x_0) \in \hat I$ (since $\hat I = \encoder(I)$), (2) for every $t<N$, $A^t\encoder(x_0)\notin \hat F$ (since $\encoder = \decoder^{-1}$, $\encoder(\decoder(A^t\encoder(x_0))) = A^t\encoder(x_0)$ and $\hat F = \encoder(F)$), and (3) $A^N\encoder(x_0)\in \hat F$ (since $\hat F =\encoder(F)$). Therefore, $\tau_0^\Sigma(x_0, F) = \tau_0^{\hat\Sigma}(\hat x_0, \hat F)$, where $\hat x_0 = \encoder(x_0)$. Since the arguments above hold for every $x_0 \in I$, we can conclude that $\tau_0^\Sigma(I, F) = \tau_0^{\hat\Sigma}(\hat I, \hat F)$.
\end{proof}
\begin{corollary}[Metastability under perfect representation]\label{col:perfect_representation}
Let $\Sigma=(S,f)$ be $(T, 0, \Dinf)$-metastable with respect to the (compact and connected) sets $D_1,\cdots,D_K$. Further, let $\M=(\encoder,\decoder, A)$ be a perfect representation for $\Sigma$, and $\hat D_i \coloneqq \encoder(D_i)$ for every $1\leq i \leq K$. Then, $\hat\Sigma=(\Sh,A)$ with $\Sh=\encoder(S)$ is $(T,0,\Dinf)$-metastable with respect to the sets $\hat D_1,\cdots,\hat D_K$.
\end{corollary}
\begin{proof}
Based on Lem.~\ref{lem:perfect_representation}, we have for every $1\leq i \leq K$ that $\tau_0^\Sigma (D_i, D_\infty)=\tau_0^{\hat\Sigma} (\hat D_i, \hat D_\infty)$. 
This proves the claim. 
\end{proof}

So far, we have discussed that, when the representation is perfect, the metastability properties of the nonlinear dynamical system carry over directly to those of the linear mapping in the latent space. However, in practice, learned representations may not be perfect. 


The following definition formalizes a useful performance measure for learning that we will use later.

\begin{definition}[$(S,T,\varepsilon)$-close learned Koopman autoencoder]
    Given an original dynamical system $\Sigma=(S,f)$ and a Koopman autoencoder model $\M=(\encoder,\decoder,A)$, for every $t\in \nats_0$ let $x_t$ and $\hat x_t$, respectively, denote the state vectors of $\Sigma$ and $\hat\Sigma=(\encoder(S),A)$, where $\hat x_0=\encoder(x_0)$. Further let $x_t'=\decoder(\hat x_t)$ for $t\in \nats_0$. 
    For a time horizon $T\in \nats_0$, we say $\M$ is an $(S,T,\varepsilon)$-close learned linear dynamics for $\Sigma$ if 
    \begin{equation}
        \forall x_0 \in S, \forall 0\leq t \leq T\colon \|x_t-x_t'\|\leq \varepsilon.
    \end{equation}
\end{definition}

In the following lemma, we show how $(S,T,\varepsilon)$-closeness affects the time scales.

\begin{lemma}[Time scales changes under  $(S,T,\varepsilon)$-close learned Koopman autoencoders]\label{lem:imperfect_time_scale_transfer}

Let $I,F\subset S$ be compact and connected sets, and let $\M=(\encoder,\decoder, A)$ be a $(S,T,\varepsilon)$-close learned representation of $\Sigma=(S,f)$. Further, let $\hat I = \encoder(I)$ and $\hat F = \encoder(F)$. If $\tau_{0}^\Sigma (I, F) \leq T$, we have  
\begin{equation}
\label{eq:timescale-reln}
    \tau_{\varepsilon}^\Sigma(I, F) \leq \tau_{0}^{\sigmahat} (\hat{I},\hat{F}).    
\end{equation}
\end{lemma}

\begin{proof}

  For $x_0 \in I$, the evolution of $x_0$ after $t$ steps in the original space and latent space is given by  $x_t = f^{t}(x_0)$ and $x_t' = \decoder(A^t \encoder(x_0)) $, respectively.
  Using $(S,T,\varepsilon)$-closeness, we have $\|x_t-x_t'\|\leq \varepsilon$ for every $1\leq t \leq T$. Therefore, we have (1) $x_t\in I^{-\varepsilon}$ implies that $\hat x_t\in \hat I$ and (2) $x_t\notin F^{+\varepsilon}$ implies that $\hat x_t\notin \hat F$. Note that $x_t\in F^{+\varepsilon}$ does \emph{not} imply that $\hat x_t\in \hat F$. Therefore, we have $\tau_{\varepsilon}^\Sigma(I, F) \leq \tau_{0}^{\sigmahat} (\hat{I},\hat{F})$.  
\end{proof}
Based on the above lemma, under the $(S,T,\varepsilon)$-closeness, large time scales in the latent space implies large time scales in the original state space with tolerance $\varepsilon$.

\begin{corollary}[Metastability under imperfect representation]\label{cor:metastability_transfer_imperfect}
Let $\M=(\encoder,\decoder, A)$ denote a Koopman autoencoder model that is $(S,T,\varepsilon)$-close with respect to the dynamical system $\Sigma=(S,f)$. Further, let $D_1, D_2,\dots, D_K$ denote the candidate sets in $S$ and $\hat D_1, \hat D_2, \dots, \hat D_K$ are defined as $\hat D_i = \encoder(D_i)$ for every $1\leq i \leq K$. Finally, let $D_\infty = \ball_\delta(x_\infty)$ and $\hat D_\infty = \encoder(D_\infty)$. We have that $(T,\varepsilon,\Dinf)$-metastability of $\Sigma$ over the original space  with respect to the candidate sets $D_1, D_2,\dots, D_K$ implies $(T, 0, \hat D_\infty)$-metastability of $\sigmahat =(\Sh, A)$ over the latent space $\Sh=\encoder(S)$ with respect to the candidate sets $\hat D_1, \hat D_2,\dots, \hat D_K$.

\end{corollary}
\begin{proof}
From Lem.~\ref{lem:imperfect_time_scale_transfer}, we have $\tau_{\varepsilon}^\Sigma(D_i, \Dinf) \leq \tau_{0}^{\sigmahat} (\hat{D}_i,\hat{D}_\infty)$ for every $1\leq i \leq K$. Therefore, $(T,\varepsilon,\Dinf)$-metastability of $\Sigma =(S, f)$ over the original space  with respect to the candidate sets $D_1, D_2,\dots, D_K$ implies $(T, 0, \hat D_\infty)$-metastability of $\sigmahat =(\Sh, A)$ over the latent space $\Sh=\encoder(S)$ with respect to the candidate sets $\hat D_1, \hat D_2,\dots, \hat D_K$. 

\end{proof}

At this stage, we have established that for a Markov process $(X_t)_{t\in\nats_0}$ defined over a measurable space $(S,\mathcal B(S))$ with average dynamics represented by the dynamical system $\Sigma=(S,f)$, metastability analysis can be carried out by (1) learning a high-fidelity Koopman autoencoder $\M=(\encoder,\decoder,A)$, (2) analyzing the time scales of the associated linear system $\sigmahat=(\Sh,A)$ with $\Sh=\encoder(S)$, and (3) applying Lem.~\ref{lem:imperfect_time_scale_transfer} and Cor.~\ref{cor:metastability_transfer_imperfect} to infer the time scales and metastability of $\Sigma$. 

However, the main motivation for employing the Koopman operator model $\M=(\encoder,\decoder,A)$ is to obtain a computationally tractable means for \emph{quantitative} estimation of time scales by leveraging the spectral properties of the linear operator $A$.

It is well known that the spectral gap of the matrix $A$ provides valuable insight into the characteristic time scales of the system. In particular, the following corollary holds. 

\begin{corollary}[Qualitative relationship between spectral radius and metastability]\label{col:qual_metastability_spectral_gap_rel}
    Given a high-fidelity Koopman model $\M=(\encoder,\decoder,A)$ learned for a nonlinear dynamical system $\Sigma=(S,f)$, the metastability properties of $\Sigma$ can be inferred from the spectral radius of $A$.
\end{corollary}
The result stated in Corollary~\ref{col:qual_metastability_spectral_gap_rel} is qualitative. Intuitively, this means that long time scales of $f$ correspond to $A$ having a spectral radius close to one ($\rho(A)\approx 1$). 

\begin{remark}
As shown in Lem.~\ref{lem:expectation_convergence}, when $\Sigma=(S,f)$ describes the evolution of an ergodic Markov process, it takes a unique steady state value. Therefore, we only consider stable dynamics over the latent space, that is $\hat\Sigma=(\Sh,A)$ is a stable linear dynamical system.
\end{remark}

The metastability characterization in Def.~\ref{def:metastability_def_main} is suitable for the case where we have a number of candidate metastable sets together with a threshold for large time scales and would like to assess metastability with respect to them. In practice, we may not have access to the candidate sets and prefer an alternative characterization that tells us whether converging to the equilibrium is fast or not: such characterization basically can be considered as an approximation of \emph{mixing time} for the underlying system. To this end, we use \emph{$\delta$-settling time}, and notice that for every choice of candidate sets $D_1, D_2,\dots,D_K$, and $0<\varepsilon\leq \delta$, we have
\begin{equation}\label{eq:comparison_delta_settling_time}
    T_{\delta}^\Sigma\geq \tau_\varepsilon^\Sigma(D_i,D_\infty).
\end{equation}
Note that $\tau_\varepsilon^\Sigma(D_i,D_\infty)$ takes the effect of $\delta$ since $D_\infty$ is defined with respect to $\delta$. Therefore, if $\delta$-settling time for a system is large, one can expect existence of candidate set(s) starting from which it takes a large amount of time to reach to $D_\infty$.  

The following lemma relates the settling time of a linear matrix $A$ to its spectral radius as well as diameter of the state space.

\begin{lemma}[Quantitative relation between spectral radius and time scales for linear stable systems]\label{lem:settling_time_lb_lin}
    Let $\hat\Sigma=(\Sh, A)$ be a stable diagonalizable linear dynamical system, where $\Sh\subset\reals^m$ is a compact and connected set. Let $\lambda_1,\lambda_2,\dots,\lambda_m$ denote the eigenvalues of $A$, ordered such that $|\lambda_m|\leq |\lambda_{m-1}|\leq \dots \leq|\lambda_1|$. 
    Then we have
    \begin{align}
        T_{\hat\delta}^{\hat\Sigma}(A,\Sh) \leq
        \frac{\log(\kappa_\infty D(\Sh)/\hat\delta)}{-\log(\rho(A))}.
    \end{align}
\end{lemma}
\begin{proof}
Since $A$ is stable, all its eigenvalues satisfy $|\lambda_i|<1$ for $1\le i\le m$, and hence $\rho(A)=|\lambda_1|$. Using Eq.~\eqref{eq:absolute_diameter}, we get that for every $x_0 \in \Sh$
 \begin{align*}
 &\|A^tx_0\|_\infty\leq \|x_0\|\|A^t\|_\infty\leq D(\Sh)\|A^t \|_\infty
 \leq \kappa_\infty D(\Sh)|\lambda_1|^t.
 \end{align*}
From the above, it suffices to ensure that $\kappa_\infty D(\Sh)|\lambda_1|^{t} \leq \hat\delta$.
We obtain $t \geq \frac{\log(\kappa_\infty D(\Sh)/\hat\delta)}{-\log(\rho(A))}$. Hence, the claim follows.
\end{proof}


Now that we know how knowledge of $\rho(A)$ in stable linear dynamical systems, we would like to have a \emph{quantitative} result that relates the time scales of $f$ and $A$. To that end, we are going to exploit characteristics of the encoding and decoding transformations. 
\begin{lemma}
Let $\M=(\encoder,\decoder,A)$ be a perfect representation for a dynamical system $\Sigma=(S,f)$. Let $L_e=U^{\encoder}$ and $L_e'=C^{\encoder}$. 
Then, we have 
\begin{equation}
\label{eq:ubound}
T_\delta^\Sigma\leq \frac{\log(2\kappa_\infty L_eD(S)/L_e'\delta)}{-\log(\rho(A))}.
\end{equation}
\end{lemma}
\begin{proof}
    Let $\hat D_\infty =\encoder(D_\infty)$. We define 
     \begin{align*}
     \hat\delta &= \sup_{\hat x, \hat x' \in \encoder(D_\infty)}\|\hat x - \hat x' \|_\infty
     =\sup_{ x,x' \in D_\infty}\|\encoder(x) - \encoder(x') \|_\infty,
     \end{align*}
    which gives
     \[
     L_e' \delta\leq\hat\delta \leq 2L_e \delta.
     \]
    We know that under perfect representation regime, 
    we have
    \[
    T_{\delta}^\Sigma\leq T_{\hat\delta/2}^{\hat\Sigma}.
    \]
      
     
     Note that under perfect representation, we have $0<L_e'\leq L_e$ (since $\decoder = \encoder^{-1}$). 
    We can write
    \begin{align*}
        T_\delta^\Sigma &\leq  T_{\hat \delta/2}^{\hat \Sigma}
        \leq \frac{\log(2\kappa_\infty D(\Sh)/\hat\delta)}{-\log(\rho(A))}
        \nonumber\\&
        \leq \frac{\log(2\kappa_\infty D(\Sh)/L_e'\delta)}{-\log(\rho(A))}
        \nonumber\\&
        \leq \frac{\log(2\kappa_\infty L_eD(S)/L_e'\delta)}{-\log(\rho(A))}.
    \end{align*} 
\end{proof}
Finally, we consider the case of imperfect representation.

\begin{theorem}\label{thm:delta_settling_time_imperfect_rep}
    Let $\M=(\encoder,\decoder,A)$ be a $(S,T,\varepsilon)$-close representation for a dynamical system $\Sigma=(S,f)$. Let $L_e=U^{\encoder}$ and $L_e'=C^{\encoder}$. We have 
    \begin{equation}
        T_\delta^\Sigma \leq \frac{\log(2\kappa_\infty L_eD(S)/L_e'(\delta-\varepsilon))}{-\log(\rho(A))}.
    \end{equation}
\end{theorem}
\begin{proof}
    Reaching to $\delta-\varepsilon$ neighborhood of the equilibrium of $f$ in $S$ by $\M$ guarantees reaching to its $\delta$ neighborhood by $\Sigma$. Further, we have $T_\delta^\Sigma \leq T_{\delta-\varepsilon}^\Sigma$. Substituting $\delta$ with $\delta-\varepsilon$ in \eqref{eq:ubound} completes the proof. 
\end{proof}

\section{Evaluation}\label{sec:experiments}

\subsection{Experimental Setup}

All the experiments were performed on Intel(R) Core(TM) i5-6600 CPU @ 3.30GHz Debian platform.  
For the purpose of these experiments, we developed a custom simulator that models a nonlinear queuing system with retries and timeouts, denoted as $\mathcal{S}$. The simulator is defined by the tuple $\langle T_{sim},n_{runs}, \nu,\mu,n_{retries},T_{timeout} \rangle$, where $T_{sim}$ denotes the maximum duration of a single simulation run, $n_{runs}$ represents the total number of runs, $\nu$ corresponds to the job arrival rate, $\mu$ denotes the service rate, $n_{retries}$ specifies the maximum number of reattempts for processing failed jobs, and $T_{timeout}$ defines the maximum permissible waiting time for a job to complete.
The configuration used for our experiments is $\langle T_{sim}=10000,n_{runs}=100, \nu=9.7,\mu=10,n_{retries}=3,T_{timeout}=9 \rangle$.

\subsection{Data Generation}

 The state evolves according to stochastic arrival and service processes with time-dependent rates. 
For each configuration, we simulate trajectories for $T_{sim}$ under a fixed arrival rate $\nu$
 and service rate 
$\mu$. The system state consists of queue length and latency, i.e., $x_t = (q_t,l_t)$.

We generate $100$ simulation trajectories $\omega$ to construct the training dataset $\Omega$. Each trajectory is embedded with a history length (depth) of 
$d=1$, producing autoregressive input–output pairs for learning. 
Our goal is to learn an autoencoder based model for the simulator. 
We learn an Autoencoder model $(\encoder,\decoder, A)$ using the dataset $\Omega$.

We trained the autoencoder-based model with two input (and corresponding output) dimensions.  To ensure a small reconstruction error ($\varepsilon$), the model $A$ should have sufficient number of parameters. Hence, we employ a relatively high number of latent dimensions ($20$). Furthermore, spectral normalization is applied to the matrix $A$ to prevent numerical instability or divergence during rollouts.

\subsection{Eigenvalue Analysis}
This section presents the results for evaluation of the Koopman autoencoder model on the queuing system. The results highlight the relationship between the system’s arrival rate $\nu$, the metastable behavior, and the spectral properties of the learned Koopman operator $A$.

Figure~\ref{fig:llen} illustrates the variation in the average latency and average queue length with respect to the arrival rate $\nu$. We observe that for higher $\nu$, a retry storm could lead to the \textit{sustaining effect}, where the system dynamics transition from a stable to a metastable regime.

\begin{figure}
    \centering
    \includegraphics[width=0.685\linewidth]{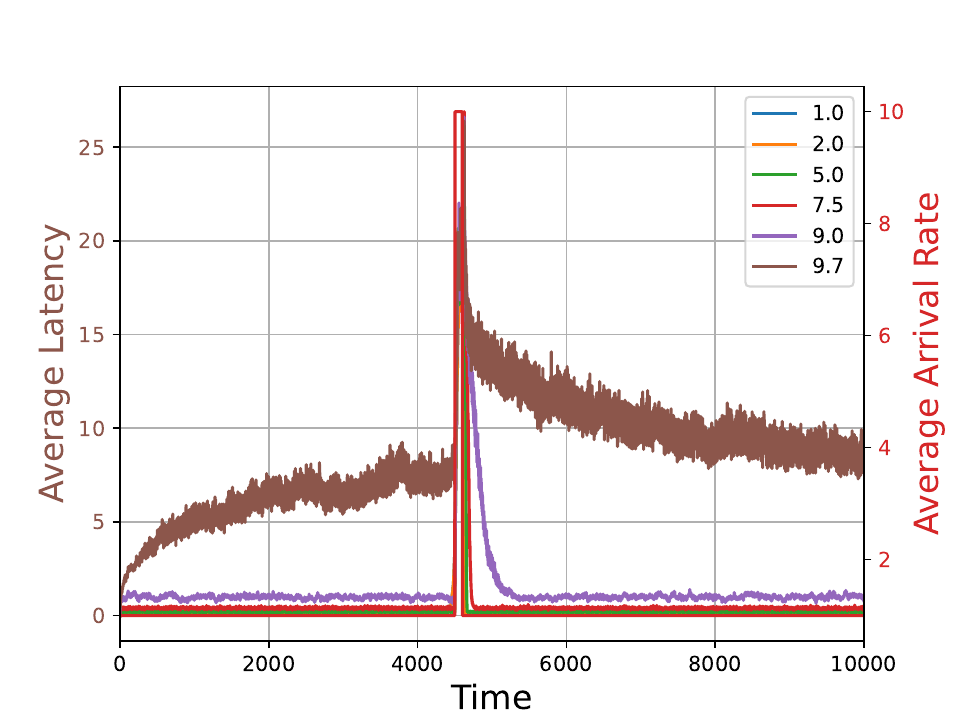}
    \includegraphics[width=0.685\linewidth]{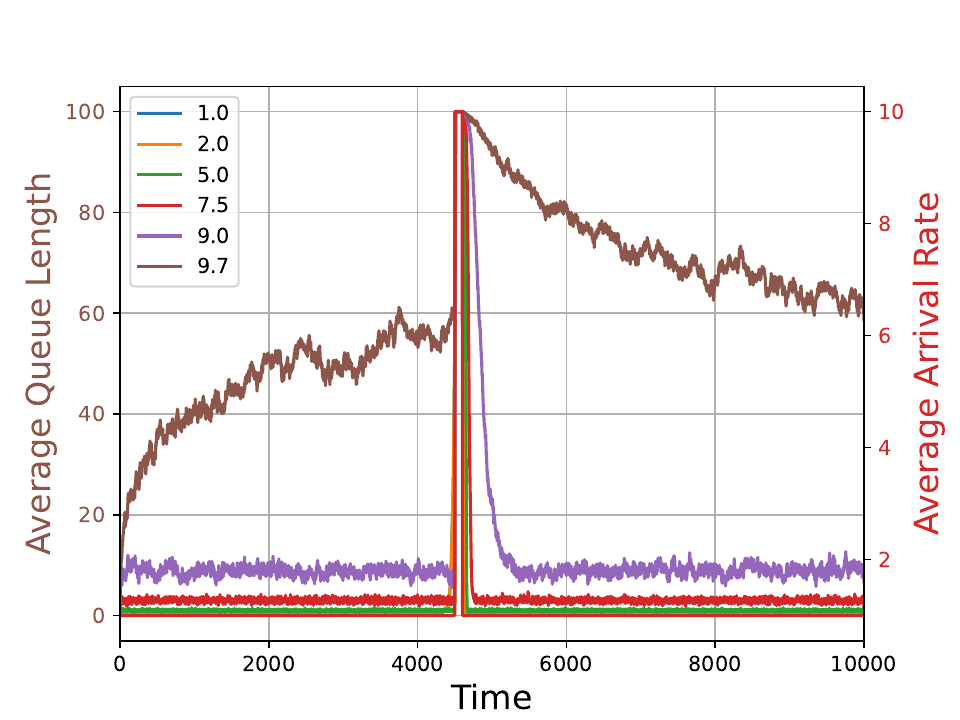}
    \caption{Average latency (top) and Average queue length (bottom) for different arrival rates.}
    \label{fig:llen}
\end{figure}


Figure~\ref{fig:learn-dyn} illustrates the predictions from the learned model, demonstrating that the learned model effectively captures the underlying system behavior.

\begin{figure}
    \centering
    \includegraphics[width=0.84\linewidth]{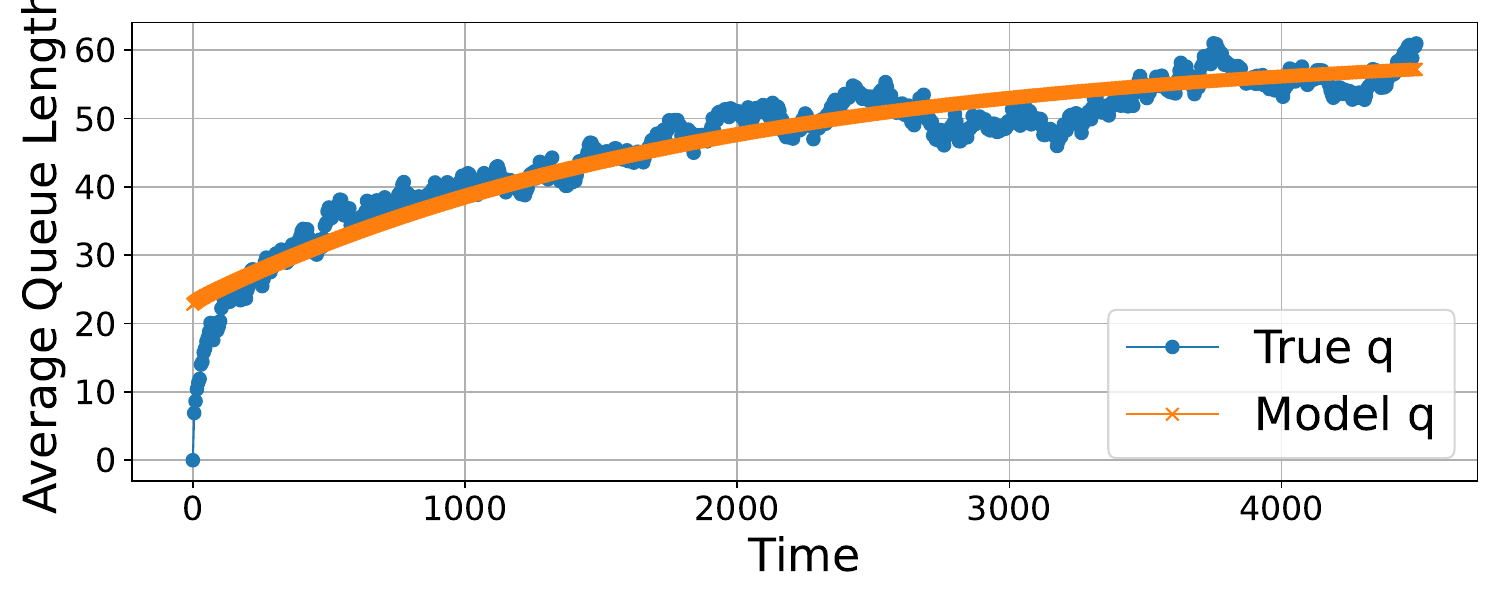}
\includegraphics[width=0.84\linewidth]{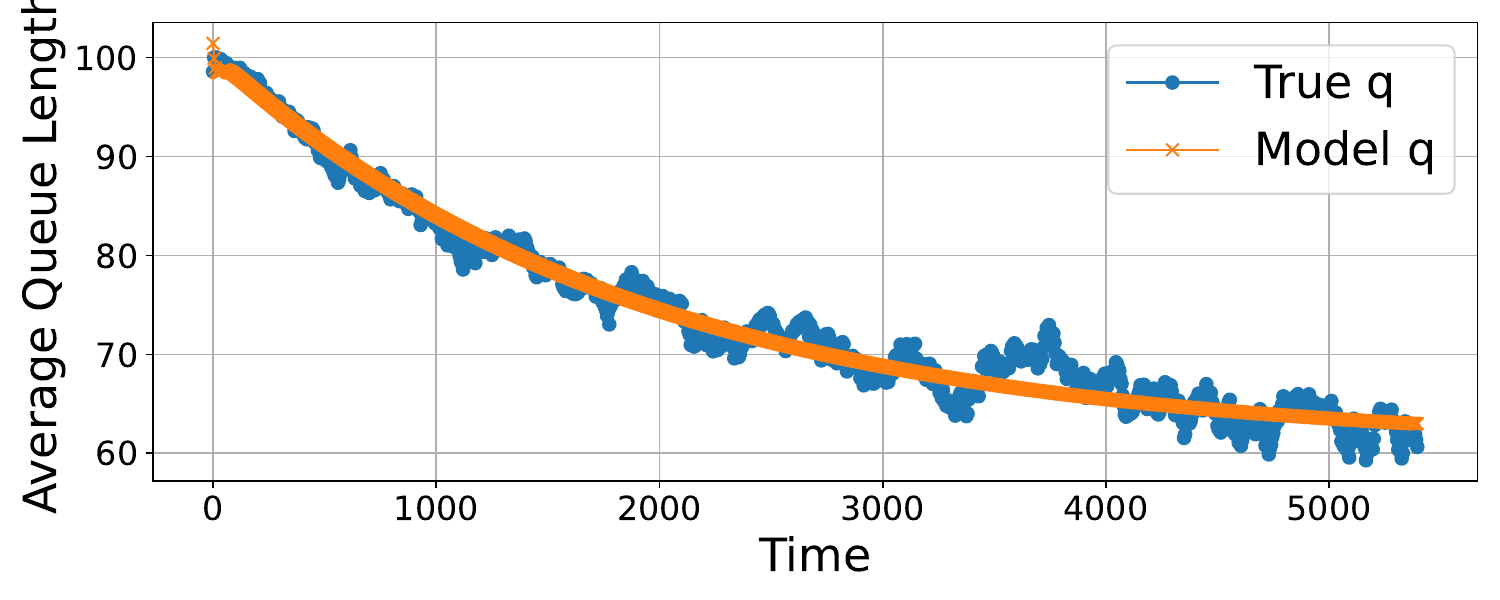}
  \caption{Illustrating the learned model performance for trajectories initialized from empty queue (top plot) and full queue (bottom plot).}
    \label{fig:learn-dyn}
    \vspace{-0.55cm}
\end{figure}
    

The largest eigenvalue of the learned Koopman matrix $A$ for different arrival rates is shown in Figure \ref{fig:eval}. We observe that the largest eigenvalue increases monotonically with the arrival rate. For low arrival rates, the dominant eigenvalue is significantly below unity, indicating fast decay toward the attractor. As $\nu$ approaches the service rate, the largest eigenvalue approaches one, reflecting the emergence of slow timescales and metastability in the system.

\begin{figure}
    \centering
    {\includegraphics[width=0.8\linewidth]{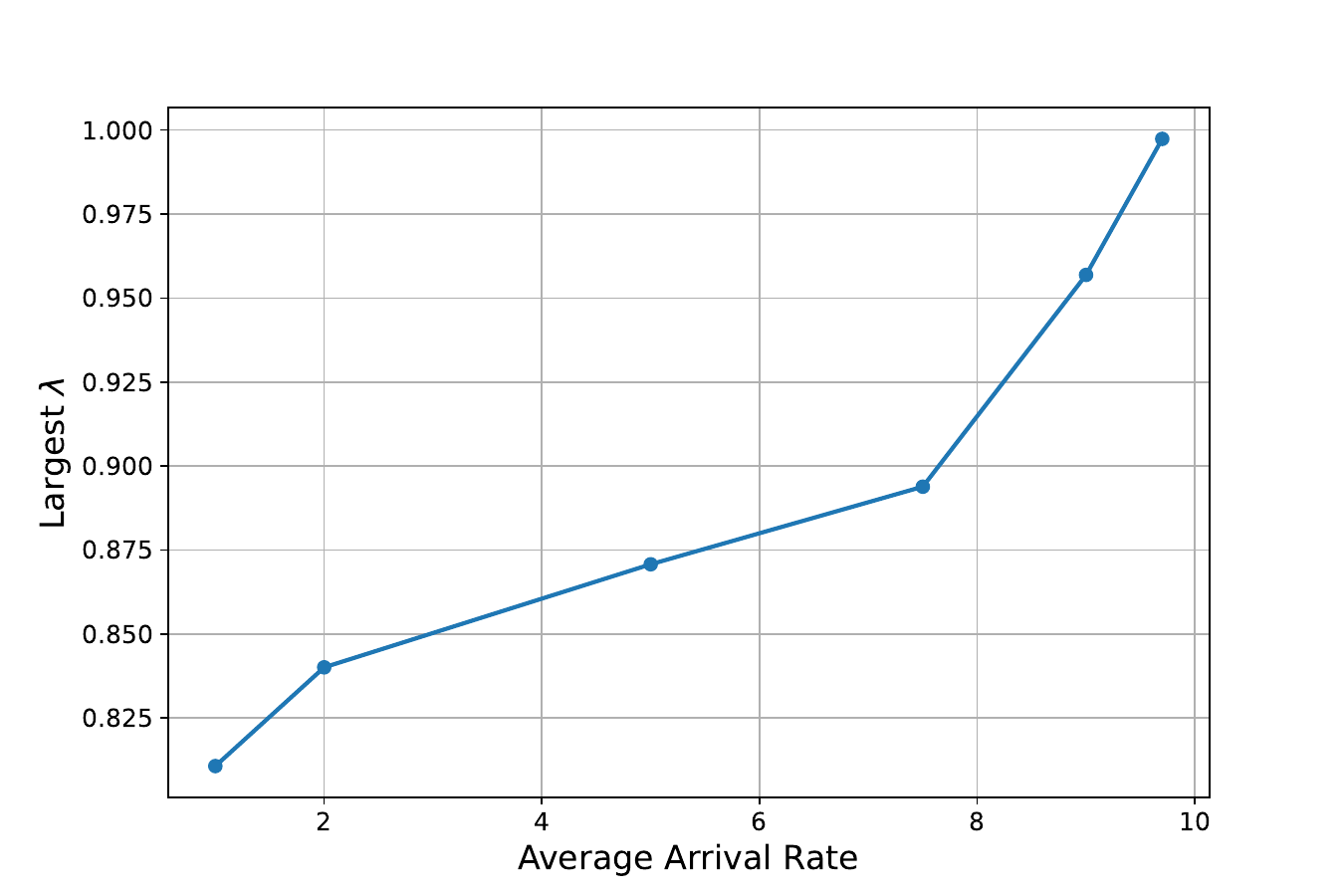}} 
    \caption{The largest eigenvalue of the Koopman matrix for different values of arrival rates.}
    \label{fig:eval}
    \vspace{-0.4cm}
\end{figure}








\subsection{Throttling}

To prevent the system from entering a metastable regime, we use a \textit{random early detection} throttling mechanism that adaptively throttles job arrivals as the queue approaches its maximum capacity. Let $q(t)$ denote the current queue length, $q_{max}$ the maximum queue capacity and $p^a$ is the admission probability. 
The throttling strategy is given as
\begin{equation}
 p^{a} =
    \begin{cases}
      1 &  q(t)<0.9q_{max}\\
      0.5 & 0.9q_{max}<q(t)\leq q_{max} \\
      0 & q(t)>q_{max}.
    \end{cases}       
\end{equation}

\begin{itemize}
    \item Normal mode: When $q(t)<0.9q_{max}$, the incoming jobs are enqueued unconditionally.
    \item Throttling mode: Once the queue occupancy exceeds $90\%$ of capacity, i.e.,  $0.9q_{max}<q(t)\leq q_{max}$, a probabilistic admission is triggered.
   Specifically, each incoming job is admitted with probability $p=0.5$ and dropped with probability $1-p$. This probabilistic intervention prevents congestion by reducing the effective arrival rate as queue gets full.
   \item Full queue: When the queue reaches full capacity, i.e., $q(t)>q_{max}$, all incoming jobs are immediately dropped.
\end{itemize}

\begin{figure}
    \centering
    {\includegraphics[width=0.684\linewidth]{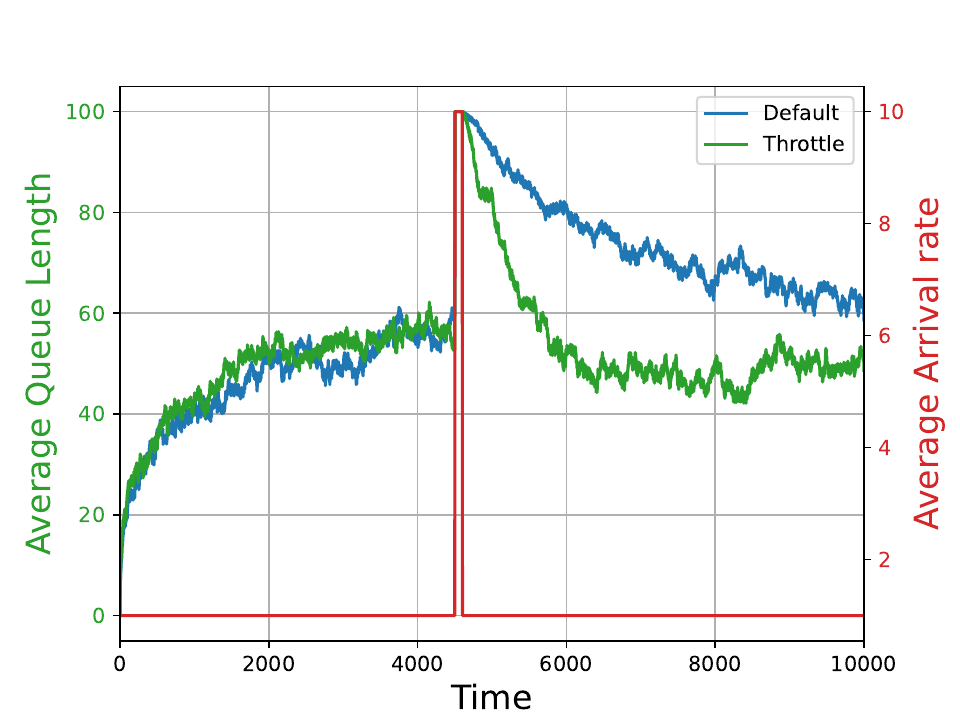}} 
    \vspace{-0.3cm}
    \caption{Illustrating the effect of throttling on the average queue length.}
    \label{fig:throt}
    \vspace{-0.4cm}
\end{figure}
In Figure~\ref{fig:throt}, we show the average queue length observed for normal simulation (default) and using the throttling strategy.

Using the trajectories obtained under throttling, we learn another autoencoder-koopman model $(\encoder^{to},\decoder^{to},A^{to})$. We observed that the largest eigenvalue for the koopman matrix $A^{to}$ if $0.9913$. Using Figure~\ref{fig:eval} and Figure~\ref{fig:llen} (bottom), we infer that throttling reduces the effective arrival rate of the system.



\subsection{LIFO Versus FIFO}

In Figure~\ref{fig:llen-lifo}, we show the comparison between latencies of Queue (FIFO) and Stack (LIFO) scheduling policies within the queuing framework.
Given that the exponential service time is memoryless, the steady-state behavior of the LIFO system is theoretically expected to be closer to that of the FIFO system. Since the newly arrived jobs see much shorter waiting times, this configuration results in higher latency variance and the potential starvation of older jobs.

In Figure~\ref{fig:qlen-lifo}, we further observe that the LIFO mechanism demonstrates superior performance relative to FIFO in recovering from the metastable state. Specifically, the largest eigenvalue of the learned Koopman operator for the LIFO-based system is $0.9915$, which is lower than the corresponding eigenvalue of $0.9973$ for the FIFO-based system, indicating a faster convergence to steady state.

\begin{figure}
    \centering
    {\includegraphics[width=0.684\linewidth]{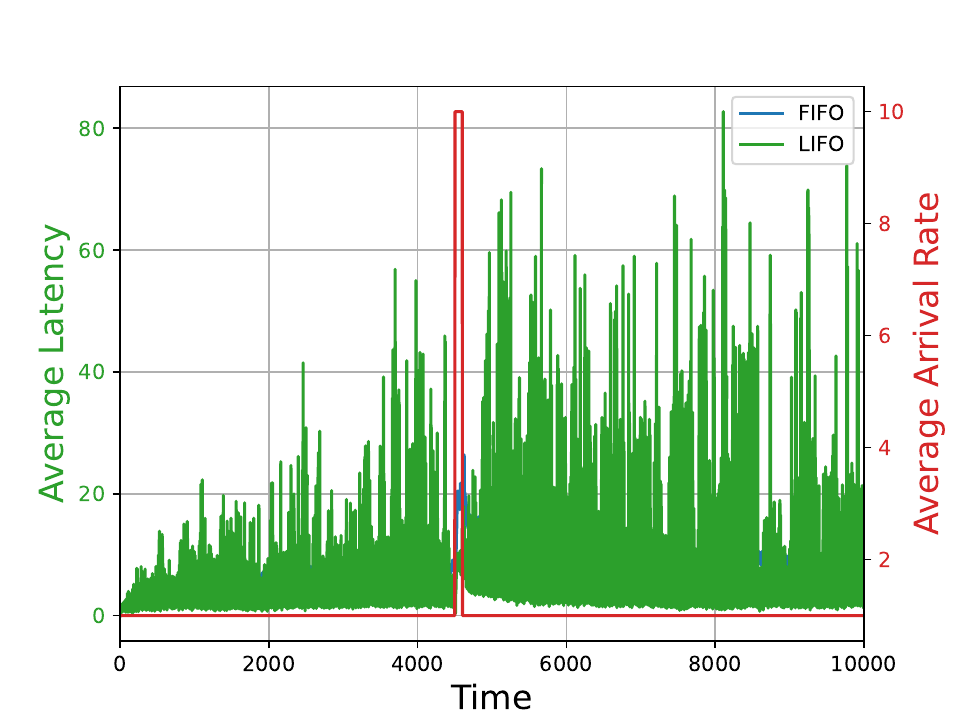}} 
    \vspace{-0.3cm}
    \caption{Comparison of average latency under LIFO and FIFO scheduling policies.}
    \label{fig:llen-lifo}
    \vspace{-0.6cm}
\end{figure}

\begin{figure}
    \centering
    {\includegraphics[width=0.684\linewidth]{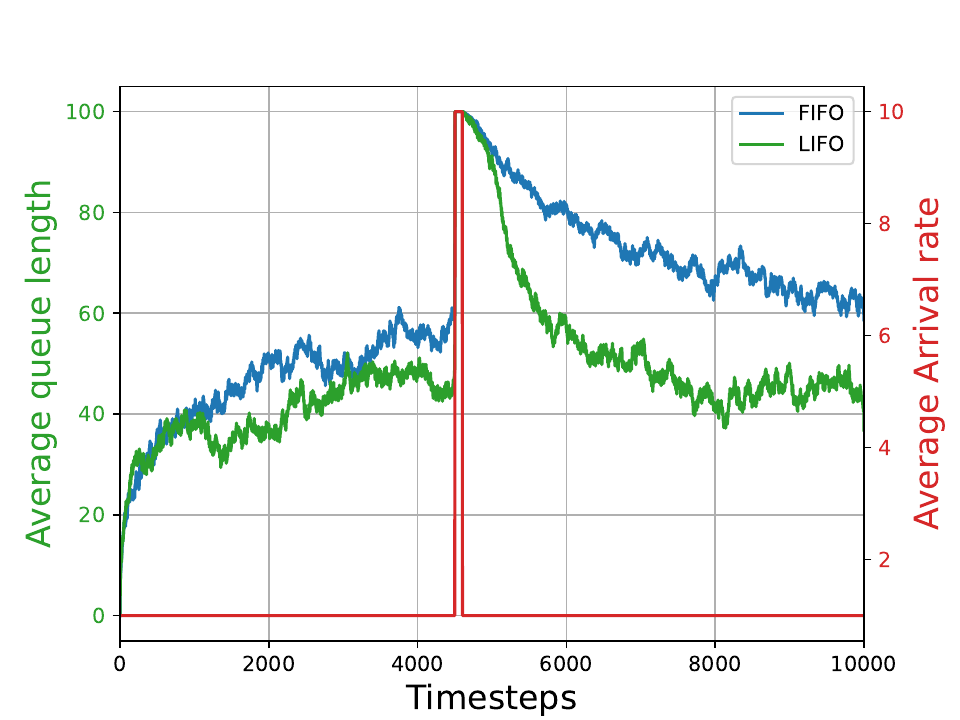}} 
    \vspace{-0.3cm}
    \caption{Comparison of average queue length under LIFO and FIFO scheduling policies.}
    \label{fig:qlen-lifo}
    \vspace{-0.2cm}
\end{figure}

\subsection{Metastability Over Shorter Simulation }

Up to this point, we trained a global model that effectively captures the metastable dynamics of the system. To further analyze metastability over shorter time scales, we learn a model over shorter duration, each corresponding to a simulation time of $1000$ timesteps, which is one-tenth of the original duration. Despite this temporal reduction, the dominant eigenvalue remains close to one ($0.9947$), indicating the preservation of metastable behavior.
Figure~\ref{fig:qlen-short} illustrates the predictions of the learned model over the short simulations.

\begin{figure}
    \centering
    {\includegraphics[width=0.84\linewidth]{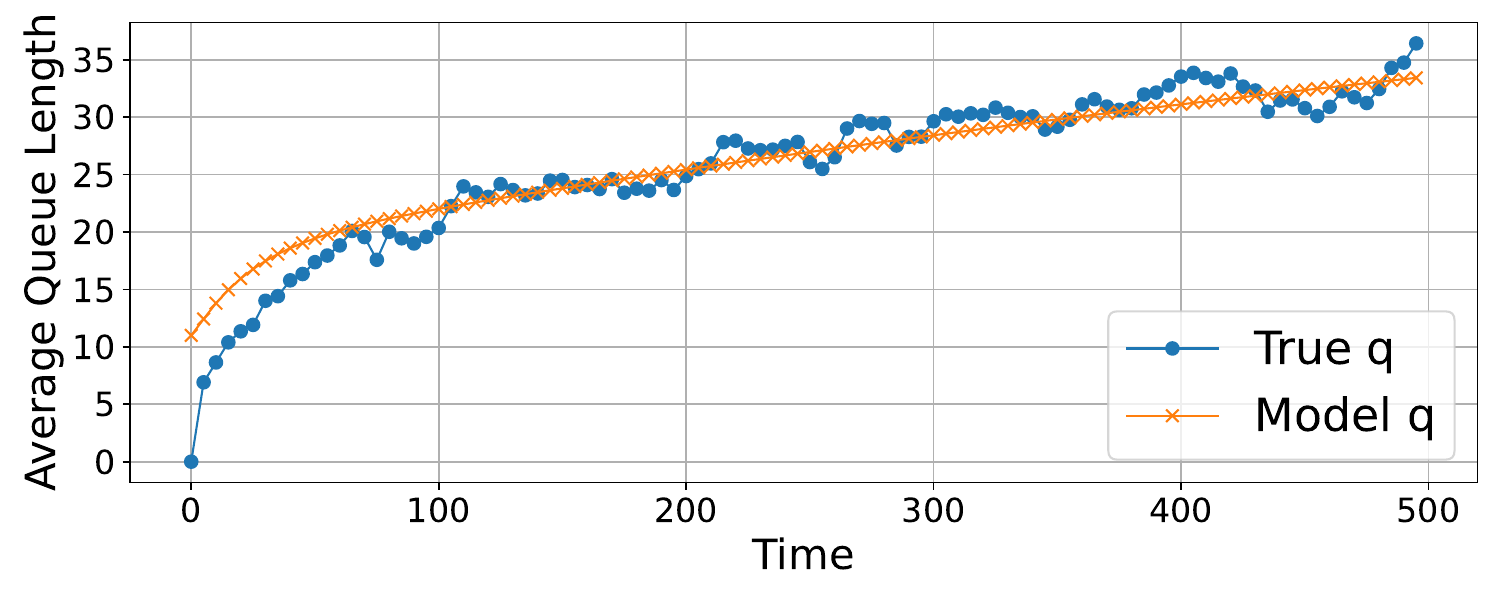}}
    \vspace{-0.3cm}
    \caption{Illustrating the performance of the learned model when trained on shorter-duration trajectories.}
    \label{fig:qlen-short}
    \vspace{-0.4cm}
\end{figure}

\subsection{Settling Times}

Figure~\ref{fig:long-sim-settle-time} (top) illustrates the comparison between the empirically observed settling times and the theoretical bounds derived from Equation~\eqref{eq:ubound} as a function of $\delta$. For these computations, we set $ L_e=1.6$, $ L_e'=0.8$, and $D(S)=103$.

\begin{figure}
    \centering
    \includegraphics[width=0.68\linewidth]{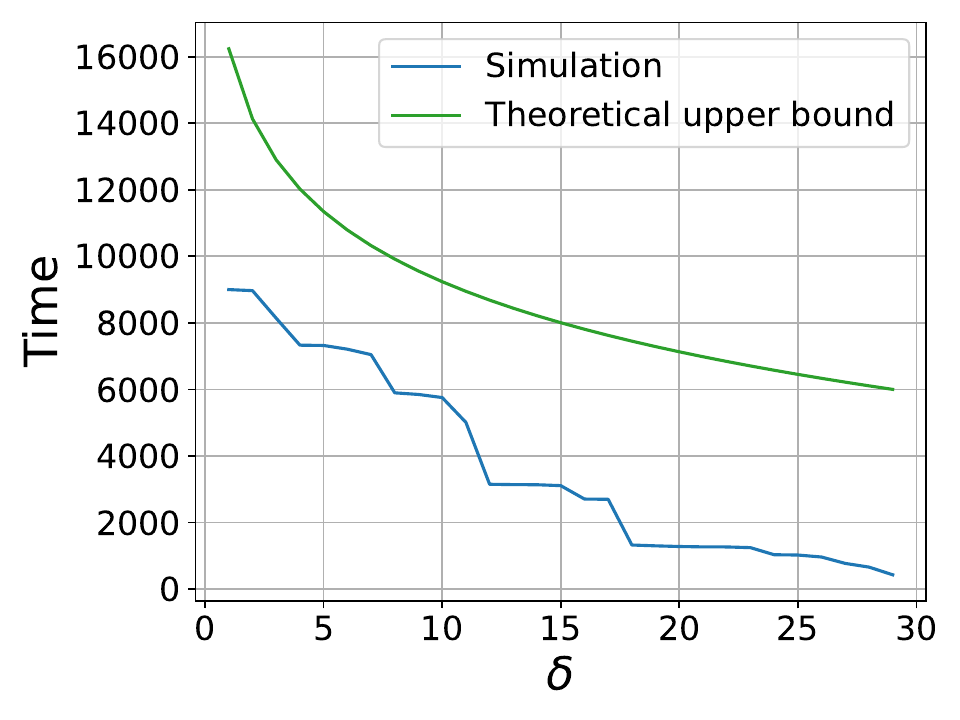} 
    \includegraphics[width=0.68\linewidth]{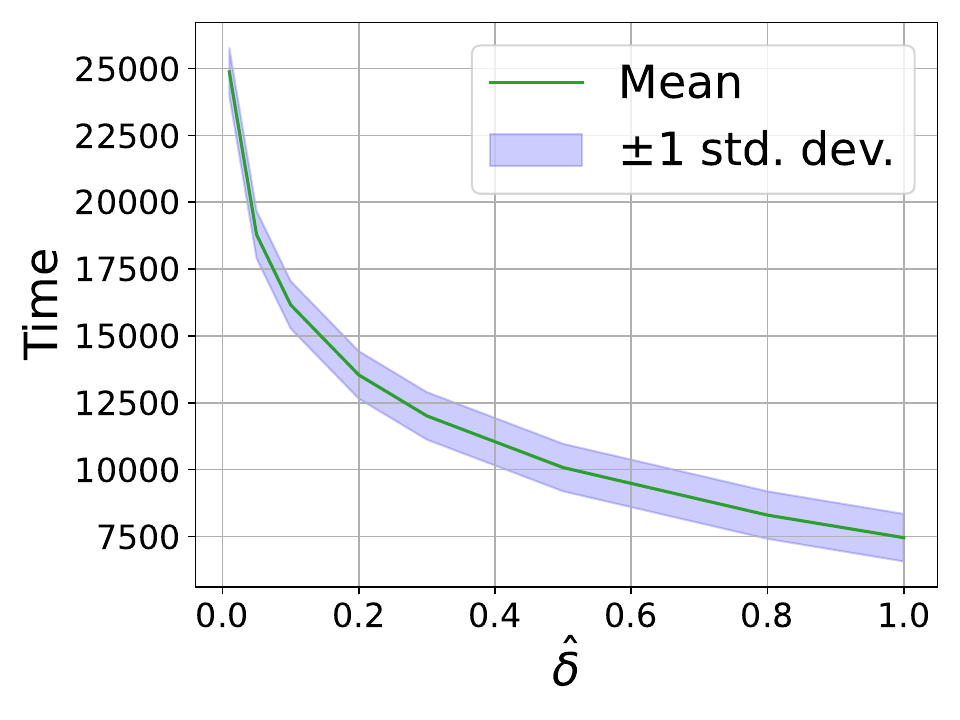}
    \vspace{-0.3cm}
    \caption{Variation in settling time in the original space as a function of $\delta$ (top plot) and in the latent space as a function of $\hat{\delta}$ (bottom plot).}
    \label{fig:long-sim-settle-time}
    \vspace{-0.4cm}
\end{figure}

For the learned model, we define the steady-state behavior as the asymptotic value obtained in the limit $t\rightarrow \infty$, as formalized in Equation~\eqref{eq:delta_settling_time_def}. In practice, since it is infeasible to simulate an infinite horizon, we approximate this limit numerically using a finite horizon of $10^5$  time steps.

Equation~\eqref{eq:delta_settling_time_def} gives the time it takes for trajectories from the learned model to get within $\hat{\delta}$ of their attractor projection.
The $\delta$-settling time  for different values of $\hat{\delta}$ is shown in Figure~\ref{fig:long-sim-settle-time} (bottom).  
 It is important to note that the error bands measured in the original space ($\delta$) differ from the corresponding tolerance levels ($\hat{\delta}$) defined in the latent space.






\subsection{Metastability in Multi Server Setting}

For simplicity, we consider a two server system in which the two servers form a \textit{tandem queuing system}. 
A job is considered complete only after it has finished execution on both servers. 
The configuration used for multi-server setting is $\langle T_{sim}=10000,n_{runs}=100, \nu=9.5,\mu=10,n_{retries}=3,T_{timeout}=9 \rangle$
and the parameters are kept identical across both servers to ensure consistency in comparison.

Figures~\ref{fig:llen-multi} and~\ref{fig:qlen-multi} illustrate the average latency and average queue lengths observed for each server in the two-server system.
From the learned system dynamics, we observe that the largest eigenvalues of the Koopman operator for servers 1 and 2 are 
$0.9832$ and 
$0.9981$, respectively. This indicates that Server 2 exhibits metastable behavior, as its dynamics evolve more slowly toward equilibrium.
This disparity in stability between the two servers can be attributed to differences in their processing times. As evident from Figure~\ref{fig:llen-multi}, the latency for Server 2 is approximately twice that for Server 1.

\begin{figure}
    \centering
    \includegraphics[width=0.684\linewidth]{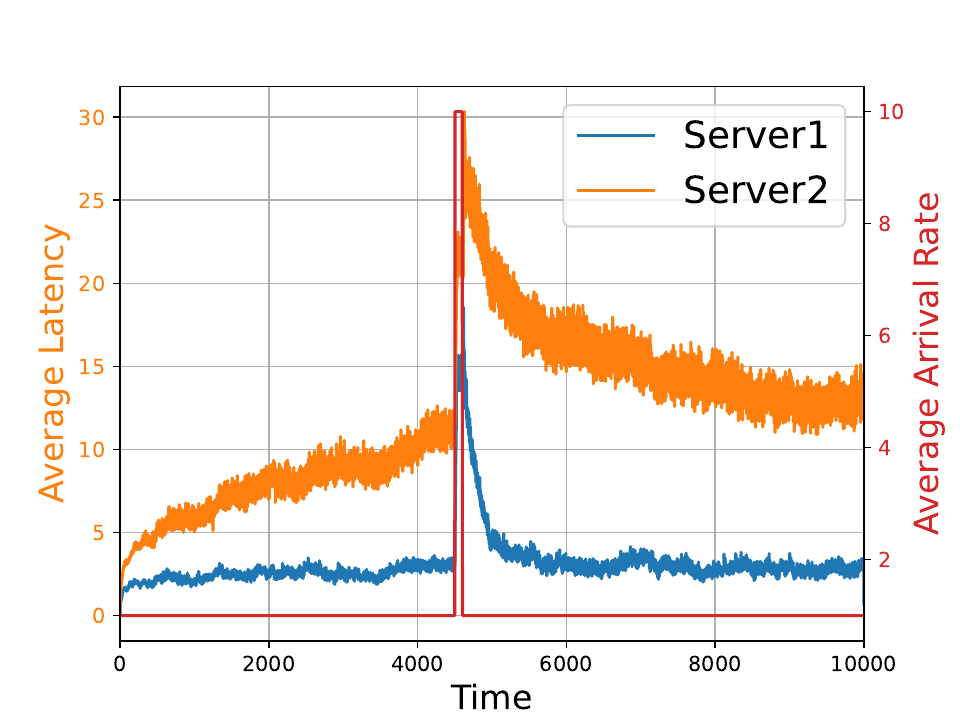} 
    \vspace{-0.3cm}
    \caption{Average latency in the two-server system.}
    \label{fig:llen-multi}
    \vspace{-0.3cm}
\end{figure}

\begin{figure}
    \centering
    \includegraphics[width=0.684\linewidth]{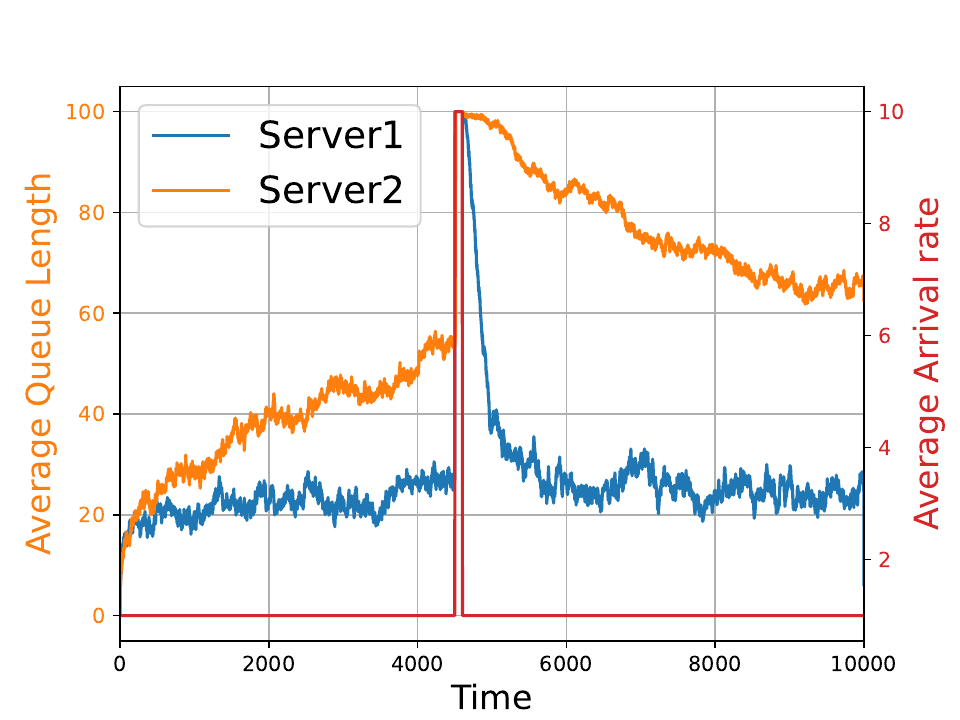}
    \vspace{-0.3cm}
    \caption{Average queue length in the two-server system.}
    \label{fig:qlen-multi}
    \vspace{-0.2cm}
\end{figure}

The $\delta$-settling times corresponding to the learned model for both servers is shown in Figure~\ref{fig:delta-times-multi}. 
Consistent with the metastable nature of Server 2, we observe larger $\delta$-settling times, indicating slower mixing and delayed convergence.

\begin{figure}
    \centering
    {\includegraphics[width=0.68\linewidth]{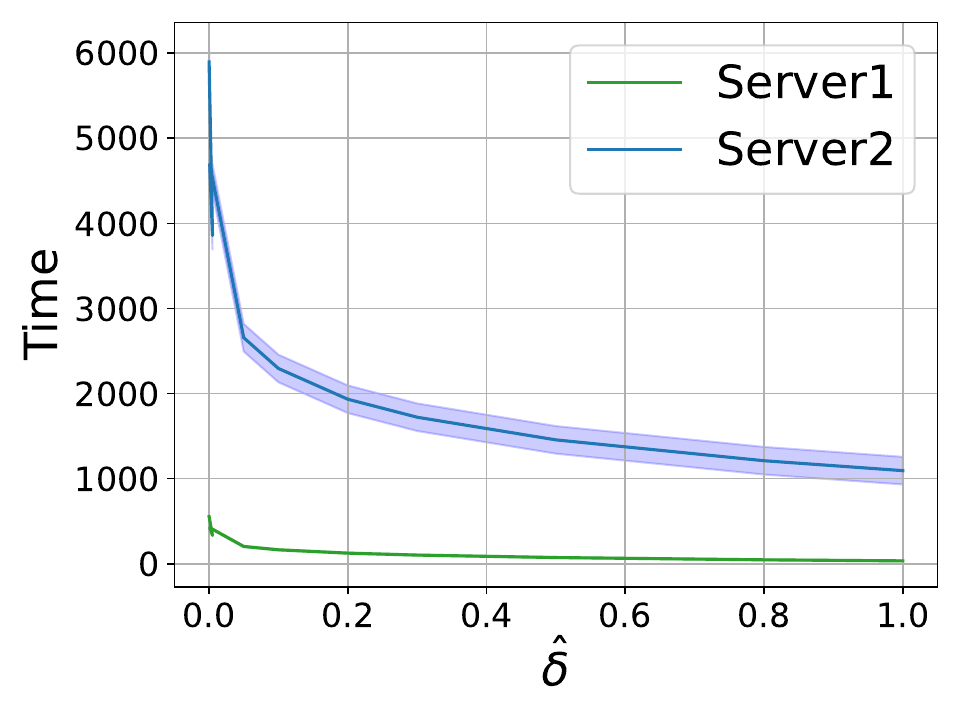}}
    \vspace{-0.6cm}
    \caption{Variation in settling time as a function of $\hat\delta$.}
    \label{fig:delta-times-multi}
    \vspace{-0.3cm}
\end{figure}

\section{Conclusion}\label{sec:conclusion}

We presented a practical framework for formal analysis of metastability in stochastic nonlinear systems with unknown dynamics. By restricting attention to ergodic stochastic systems, we showed how one can reason about metastability through the underlying deterministic dynamics that govern the evolution of state moments, such as expectations. We learned high-fidelity representations of these deterministic dynamics from finitely many system trajectories using Koopman autoencoders. Koopman autoencoders define a transformation into a latent space where the dynamics become linear. We then analyzed the spectral properties of the learned linear dynamics to reason about metastability, focusing on the eigenvalue of largest modulus. 
Empirical results demonstrated that the proposed framework effectively identifies and characterizes metastable phenomena across several distributed system case studies.
As future work, we plan to explore alternative approaches for relating stochastic dynamics to deterministic representations using moment closures~\cite{Naghnaeian17}.


\bibliographystyle{ACM-Reference-Format}
	\bibliography{references}
	

\appendix
\end{document}